\documentclass[letterpaper, 10 pt, conference]{ieeeconf}

\usepackage{cite}
\usepackage{xcolor}
\usepackage{graphicx}
\usepackage{amsmath}
\usepackage{array}
\usepackage{multirow}
\usepackage{booktabs}
\usepackage{caption}
\usepackage{siunitx}

\usepackage{amsthm}
\newtheorem{lemma}{Lemma}

\usepackage{amsmath,amssymb,amsfonts}
\IEEEoverridecommandlockouts
\title{\LARGE \bf
Contraction Certification from Streaming Data: Wasserstein
Robustness and Compositional Stability for Interconnected
Nonlinear Systems}

\author{Faegheh Moazeni$^{1,*}$%
\thanks{}%
\thanks{The author is with the Civil and Environmental Engineering
Department at Lehigh University, Bethlehem, PA 18015, USA.
Email: {\tt\small fam321@lehigh.edu}}
\thanks{*Corresponding author.}}

\begin{document}

\maketitle
\thispagestyle{empty}
\pagestyle{empty}

\begin{abstract}
Streaming contraction certificates, which determine in real time
whether observed data is sufficient to certify a safe control
action, face two structural challenges: the disturbance
distribution shifts during operation, and the system consists
of coupled subsystems whose joint model is unavailable. This
paper addresses both. First, we develop a Wasserstein-robust
certificate $\beta_{\mathrm{cert}}(t,\varepsilon) =
\hat{\beta}(t) - \rho(t)(1+2\varepsilon(t))$, where
$\varepsilon(t)$ is estimated online from the empirical excess
kurtosis of recent residuals, so the certificate degrades
gracefully under distributional shift rather than failing
catastrophically. Second, we prove that local certificates
$\beta_A$ and $\beta_B$, estimated independently from each
subsystem's data, compose into a network-level guarantee via
$\beta_{\mathrm{net}} = \tfrac{\beta_A+\beta_B}{2} -
\sqrt{\tfrac{(\beta_A-\beta_B)^2}{4}+\gamma^2} > 0$
whenever $\gamma < \sqrt{\beta_A\cdot\beta_B}$, with no
joint model required. On a five-node G5 benchmark under
three noise regimes, i.e., Gaussian, heavy-tailed Laplace, and
spike events, the Wasserstein certificate remains valid in
$73\%$ of spike-regime timesteps versus $33\%$ for the
Gaussian baseline ($2.2\times$ improvement), while the
Gaussian certificate never authorizes deployment during the
spike window. The compositional framework correctly identifies
all three coupling regimes from local data alone, with
$\gamma_{\mathrm{warn}} = \sqrt{\beta_A\cdot\beta_B}
\approx 0.98$ precisely predicting network-level contraction
loss.
\end{abstract}

\section{Introduction}
\label{sec:intro}

Real-time certified control of nonlinear networked systems
requires answering, at every timestep, whether the data
collected so far is sufficient to guarantee that the next
control action will not destabilize the system. A scalar
streaming contraction certificate $\beta_{\mathrm{cert}}(t)
= \hat{\beta}(t) - \rho(t)$ can answer this question
recursively from partial observations: $\hat{\beta}(t)$ is
estimated via integral regression on a sliding window of
input-output data, and $\rho(t)$ is a data-dependent
uncertainty radius that bounds the estimation error in the
closed-loop Jacobian \cite{lohmiller1998contraction}. The
controller is deployed the moment $\beta_{\mathrm{cert}}(t)$
crosses and sustains above a positive margin. This base
framework, however, rests on two assumptions that break in
every real networked system of consequence: the disturbance
distribution is fixed and known, and each subsystem operates
in isolation. This paper removes both assumptions
simultaneously.

\subsection*{Related Work}

\textbf{Distributional robustness in data-driven control.}
Existing data-driven certificates handle uncertainty in one
of two ways. Bounded-noise formulations
\cite{de2019formulas,van2020data} bound the disturbance
magnitude and derive worst-case certificates, but become
vacuous when bounds are large and provide no mechanism for
adapting to shifts in the noise character during operation.
Gaussian-assumption certificates are tighter under nominal
conditions but fail without warning when the realized
distribution has heavier tails than assumed; a structural
guarantee of failure in any system subject to demand surges,
sensor anomalies, or rare extreme events. Distributionally
robust optimization \cite{rahimian2022frameworks} and
Wasserstein ambiguity sets \cite{mohajerin2018data} have
been developed in the stochastic programming community to
certify performance over a family of distributions within a
ball of radius $\varepsilon$ around a nominal in the
Wasserstein-2 metric $\mathcal{W}_2$. These tools have been
applied to robust model predictive control
\cite{coulson2021distributionally} and chance-constrained
optimization, but remain offline batch operations: the
ambiguity radius $\varepsilon$ is a designer-specified
constant, not a quantity estimated recursively from streaming
residuals. To the best of our knowledge, no existing framework connects Wasserstein
robustness to a real-time streaming certificate that adapts
$\varepsilon$ as the disturbance distribution shifts.

\textbf{Compositional stability from data.}
Classical small-gain theory \cite{zames1966input} and
dissipativity-based composition \cite{willems1972dissipative}
establish that a network of stable, dissipative subsystems
remains stable under bounded coupling, provided a gain
condition holds between subsystems. Passivity-based
distributed control \cite{arcak2007passivity} and
contraction-based composition \cite{lohmiller1998contraction}
extend these ideas to nonlinear systems, with the latter
showing that two systems contracting at rates $\beta_A$ and
$\beta_B$ under coupling gain $\gamma$ remain contracting
when $\gamma < \sqrt{\beta_A \cdot \beta_B}$. All of these
results, however, require exact subsystem models or
storage functions derived from first principles. On the
data-driven side,  \cite{de2019formulas}
synthesize stabilizing controllers for individual subsystems
from data, and dissipativity properties have been estimated
from data for linear systems \cite{van2023informativity}.
Distributed DeePC \cite{kohler2022data} partitions the
Hankel matrix by subsystem but requires a jointly collected
global dataset and provides no stability certificate for the
coupled nonlinear network. To the best of our knowledge, no existing result composes
streaming local contraction certificates, estimated from
each subsystem's own data during disconnected
operation, into a network-level stability guarantee.

\subsection*{Contributions}

This paper makes the following contributions:

\begin{itemize}

\item \textbf{Wasserstein-robust streaming certificate.}
We develop $\beta_{\mathrm{cert}}(t,\varepsilon) =
\hat{\beta}(t) - \rho(t)(1+2\varepsilon(t))$, a streaming
contraction certificate valid over the Wasserstein-2
ambiguity set $\mathcal{P}(\varepsilon) = \{P :
\mathcal{W}_2(P,P_0) \leq \varepsilon\}$, where the radius
$\varepsilon(t)$ is estimated online from the empirical
excess kurtosis of recent regression residuals. The
certificate degrades gracefully as the disturbance
distribution shifts rather than failing catastrophically,
and collapses to the nominal certificate
$\beta_{\mathrm{cert}}(t) = \hat{\beta}(t)-\rho(t)$ when
$\varepsilon(t) \to 0$.

\item \textbf{Online estimation of the ambiguity radius.}
We derive a streaming estimator for $\varepsilon(t)$ that
uses only the causal residual history of the integral
regression estimator, with no additional sensors or
distributional assumptions. The estimator correctly
identifies all three noise regimes, i.e., Gaussian, heavy-tailed,
and spike events, in a single continuous trajectory.

\item \textbf{Compositional contraction certification.}
We prove that local streaming certificates $\beta_A$ and
$\beta_B$, estimated independently from each subsystem's
partial observations, compose into a network-level
contraction guarantee via $\beta_{\mathrm{net}} =
\tfrac{\beta_A+\beta_B}{2} -
\sqrt{\tfrac{(\beta_A-\beta_B)^2}{4}+\gamma^2}$, provided
$\gamma < \sqrt{\beta_A \cdot \beta_B}$. No joint model
of the coupled system is required at any stage.
\end{itemize}

The remainder of this paper is organized as follows.
Section~\ref{sec:problem} states the problem. Section~\ref{sec:wasserstein} develops
the Wasserstein-robust certificate and the online
$\varepsilon$ estimator. Section~\ref{sec:composition}
develops the compositional certification framework.
Section~\ref{sec:results} presents simulation results.
Section~\ref{sec:conclusion} concludes.

\section{Problem Formulation}
\label{sec:problem}

\subsection{Interconnected System Model}

We study a nonlinear system comprising two coupled subsystems
$\Sigma_A$ and $\Sigma_B$, each evolving on a subgraph of a
known physical network $\mathcal{N}=(\mathcal{V},\mathcal{E})$
with node set $\mathcal{V}$ and edge set $\mathcal{E}$. Each
subsystem has the form
\begin{align}
 & \Sigma_i:\; \dot{x}^{(i)} = f^{(i)}\!\left(x^{(i)},
x^{(j)}\right) + B^{(i)}u^{(i)} + E^{(i)}\eta(t), \notag\\
&\quad\quad\quad i\in\{A,B\},\; j\neq i,
\label{eq:subsys}  
\end{align}
where $x^{(i)}\in\mathbb{R}^{n_i}$ is the local state,
$u^{(i)}\in\mathbb{R}^{m_i}$ is the local control input, and
$\eta(t)$ is an exogenous disturbance whose distribution may
change during operation. The interconnection between $\Sigma_A$
and $\Sigma_B$ is parameterized by a scalar gain $\gamma\geq 0$:
\begin{equation}
f^{(i)}\!\left(x^{(i)},x^{(j)}\right) =
h^{(i)}\!\left(x^{(i)}\right) +
\gamma\,\phi^{(i)}\!\left(x^{(j)}\right),
\label{eq:coupling}
\end{equation}
where $h^{(i)}$ captures the isolated subsystem dynamics and
$\phi^{(i)}$ captures the cross-subsystem coupling. The vector
fields $h^{(i)}$, matrices $B^{(i)}$ and $E^{(i)}$, and the
gain $\gamma$ are all unknown. Each subsystem is equipped with
a local sensor producing a partial measurement
$y^{(i)}=C^{(i)}x^{(i)}\in\mathbb{R}^{p_i}$ with $p_i\ll n_i$;
no sensor has access to the full state of either subsystem.

\subsection{Base Streaming Certificate}
\label{sec:base}

At each timestep $t_k$, the closed-loop Jacobian
$J_{\mathrm{cl}}^{(i)} = J_{\mathrm{obs}}^{(i)} +
B_{\mathrm{obs}}^{(i)}K^{(i)}$ is estimated from a sliding
window of $M$ input-output pairs via integral regression.
For each column $q=1,\ldots,M$ define
\begin{equation}
\Delta Y_q = y(t_q+h)-y(t_q),\quad
Z_q = \begin{bmatrix}
\int_{t_q}^{t_q+h} y\,d\tau \\
\int_{t_q}^{t_q+h} u\,d\tau
\end{bmatrix},
\label{eq:regressor}
\end{equation}
so that $\Delta Y \approx \Theta Z$ with $\Theta =
[J_{\mathrm{obs}}\;B_{\mathrm{obs}}]$. Integrating over $h$
steps averages out noise rather than differencing it. The
ridge-regularized estimate is $\hat{\Theta} = \Delta Y Z^\top
(ZZ^\top+\lambda I)^{-1}$, giving the contraction rate estimate
\begin{equation}
\hat{\beta}(t) = -\lambda_{\max}\!\left(\tfrac{1}{2}
\bigl(\hat{J}_{\mathrm{cl}}+\hat{J}_{\mathrm{cl}}^\top
\bigr)\right).
\label{eq:betahat}
\end{equation}
The data-dependent uncertainty radius
\begin{equation}
\rho(t) = \frac{c\,(1+\|K\|_2)\,\mathrm{RMS}(R)}
{\sqrt{\lambda_{\min}(ZZ^\top/M)+\lambda}},
\label{eq:rho}
\end{equation}
where $R=\Delta Y - \hat{\Theta}Z$ is the residual matrix and
$c>0$ is a tunable conservatism constant, provides a valid
lower bound on the true contraction rate via
$\beta_{\mathrm{cert}}(t) = \hat{\beta}(t)-\rho(t) \leq
\beta^*(t)$ \cite{lohmiller1998contraction}. The controller
is deployed when $\beta_{\mathrm{cert}}(t)\geq\beta_{\mathrm{margin}}$
for $n_s$ consecutive samples.

\subsection{Contraction and Local Certification}

A subsystem $\Sigma_i$ under local feedback
$u^{(i)}=K^{(i)}y^{(i)}$ is \textit{exponentially contracting}
at rate $\beta_i>0$ if its closed-loop Jacobian
$J^{(i)}_{\mathrm{cl}}$ satisfies
\begin{equation}
J^{(i)}_{\mathrm{cl}} +
\left(J^{(i)}_{\mathrm{cl}}\right)^\top \preceq -2\beta_i I
\quad \forall\,x^{(i)}\in\mathcal{X}_i,
\label{eq:contraction_local}
\end{equation}
which implies $\|x^{(i)}_1(t)-x^{(i)}_2(t)\| \leq
e^{-\beta_i t}\|x^{(i)}_1(0)-x^{(i)}_2(0)\|$ for any two
local trajectories under the same input
\cite{lohmiller1998contraction}. The local streaming
certificate $\beta^{(i)}_{\mathrm{cert}}(t)$, computed via
\eqref{eq:betahat}-\eqref{eq:rho} from subsystem $\Sigma_i$'s
own measurement history without any information from $\Sigma_j$,
is used in Section~\ref{sec:composition}. We address
two challenges that prevent existing streaming certificates from
certifying the coupled network~\eqref{eq:subsys}: (i) the
disturbance $\eta(t)$ is drawn from an unknown and potentially
time-varying distribution; (ii) the local certificates
$\beta^{(A)}_{\mathrm{cert}}$ and $\beta^{(B)}_{\mathrm{cert}}$
do not automatically guarantee stability of the coupled system.

\subsection{Five-Node Benchmark Network (G5)}

We validate on a five-node benchmark whose coupled structure
directly instantiates \eqref{eq:subsys}--\eqref{eq:coupling}.
Subsystem $\Sigma_A$ governs nodes 1--3:
\begin{align}
\dot{x}_1 &= -x_1\tanh(x_1)+u_1+\gamma\,x_4,\label{eq:SA1}\\
\dot{x}_2 &= x_2^3-x_2+a_{12}\,x_1,\label{eq:SA2}\\
\dot{x}_3 &= -x_3+a_{13}\,x_1 x_3+a_{32}\,x_2,\label{eq:SA3}
\end{align}
and subsystem $\Sigma_B$ governs nodes 4--5:
\begin{align}
\dot{x}_4 &= -x_4+u_4+\eta(t),\label{eq:SB4}\\
\dot{x}_5 &= -x_5+x_4^2.\label{eq:SB5}
\end{align}
Fixed parameters: $a_{12}=0.3$, $a_{13}=0.4$, $a_{32}=0.2$.
The coupling gain $\gamma$ is swept over $[0,2]$ in
Experiment~2 and held at $\gamma=0.15$ in Experiment~1.
Local measurements are $y^{(A)}=x_1$ and $y^{(B)}=x_4$;
nodes $x_2,x_3,x_5$ are unobservable. The candidate local
controllers are $u_1=k_1 x_1$ and $u_4=k_4 x_4$ with
$k_1=-2.5$ and $k_4=-3.0$.

The disturbance $\eta(t)$ entering $\Sigma_B$
through~\eqref{eq:SB4} takes three distinct stochastic
characters in Experiment~1. During $t\in[0,10]$~s, $\eta(t)$
is Gaussian with standard deviation $\sigma_0=0.08$, a
light-tailed regime under which a fixed-distribution
certificate is adequate. During $t\in[10,20]$~s, $\eta(t)$
follows a Laplace distribution with scale $0.5$, whose excess
kurtosis of $3$ produces tails substantially heavier than the
Gaussian nominal. During $t\in[20,30]$~s, $\eta(t)$ is a
Laplace variate with scale $0.3$ augmented by impulsive
events of magnitude $4.5$--$7.5\,\sigma_0$ occurring at rate
$5.5\%$ per timestep; the spike-event regime representing
pipe bursts, demand surges, or sensor anomalies. These three
regimes appear in a single uninterrupted $30$~s trajectory,
requiring a certificate that adapts without restart or
re-identification.

\section{Wasserstein-Robust Streaming Certificate}
\label{sec:wasserstein}

\subsection{Ambiguity Radius Estimator}

Let $\kappa_e(t)$ denote the empirical excess kurtosis of the
regression residuals $R$ computed over a sliding tail window
of $N_{\mathrm{tail}}$ samples. The streaming ambiguity radius
is
\begin{equation}
\varepsilon(t) = \mathrm{clip}\!\left(
\varepsilon_{\mathrm{floor}} +
\alpha_\kappa\,\max(0,\kappa_e(t)) +
\alpha_s\,\hat{s}(t),\;
\varepsilon_{\mathrm{floor}},\;
\varepsilon_{\mathrm{cap}}\right),
\label{eq:eps_est}
\end{equation}
where $\hat{s}(t)$ is the fraction of residuals in the tail
window exceeding $3\sigma_0$, and $\alpha_\kappa,\alpha_s>0$
are tunable weights. When residuals are sub-Gaussian,
$\kappa_e\approx 0$ and $\hat{s}(t)\approx 0$, so
$\varepsilon(t)\approx\varepsilon_{\mathrm{floor}}$ and the
certificate reduces to the nominal $\hat{\beta}(t)-\rho(t)$.
As tails grow heavier, $\varepsilon(t)$ rises, inflating the
penalty $\rho(t)(1+2\varepsilon(t))$ and absorbing the
increased distributional uncertainty.

\subsection{Wasserstein-Robust Certificate}

The Wasserstein-robust streaming certificate is
\begin{equation}
\beta_{\mathrm{cert}}(t,\varepsilon) =
\hat{\beta}(t) - \rho(t)(1+2\varepsilon(t)).
\label{eq:wass_cert}
\end{equation}
The inflation factor $(1+2\varepsilon)$ is motivated as
follows. For a disturbance distribution $P$ within the
Wasserstein-2 ball $\mathcal{P}(\varepsilon)=\{P:
\mathcal{W}_2(P,P_0)\leq\varepsilon\}$ around the nominal
$P_0$, the worst-case increase in regression residual
variance is bounded by $2\varepsilon\,\mathrm{RMS}(R)$
to first order in $\varepsilon$ \cite{mohajerin2018data},
so the effective uncertainty radius under the worst-case
$P\in\mathcal{P}(\varepsilon)$ is $\rho(t)(1+2\varepsilon)$.
The certificate degrades gracefully as $\varepsilon(t)$
grows rather than failing catastrophically, and collapses
to $\hat{\beta}(t)-\rho(t)$ when $\varepsilon(t)\to 0$.

\section{Compositional Contraction Certification}
\label{sec:composition}

Given local certified rates $\beta^{(A)}_{\mathrm{cert}}$ and
$\beta^{(B)}_{\mathrm{cert}}$ estimated independently from
each subsystem's streaming data, the following lemma
establishes a network-level contraction guarantee without
requiring a joint model of the coupled system.

\begin{lemma}[Compositional Contraction]
\label{lem:composition}
Let $\Sigma_A$ and $\Sigma_B$ each satisfy the local
contraction condition~\eqref{eq:contraction_local} with
certified rates $\beta_A>0$ and $\beta_B>0$ respectively,
and let the coupling satisfy~\eqref{eq:coupling} with scalar
gain $\gamma\geq 0$. Then the coupled system
\eqref{eq:subsys} is exponentially contracting at the
network rate
\begin{equation}
\beta_{\mathrm{net}} = \frac{\beta_A+\beta_B}{2} -
\sqrt{\frac{(\beta_A-\beta_B)^2}{4}+\gamma^2},
\label{eq:beta_net}
\end{equation}
which is positive if and only if
$\gamma < \sqrt{\beta_A\cdot\beta_B}$.
\end{lemma}

\begin{proof}
The symmetric part of the closed-loop Jacobian of the
coupled system~\eqref{eq:subsys} under local feedback takes
the two-block form
\[
J_s = \begin{pmatrix}
-\beta_A & \gamma \\
\gamma & -\beta_B
\end{pmatrix},
\]
where the diagonal entries follow from the local contraction
conditions~\eqref{eq:contraction_local} applied to each
subsystem, and the off-diagonal entries are bounded by
$\gamma$ from the coupling structure~\eqref{eq:coupling}.
The coupled system is contracting at rate $\beta_{\mathrm{net}}$
if and only if $J_s\preceq -\beta_{\mathrm{net}}I$, i.e.,
$\beta_{\mathrm{net}} \leq \lambda_{\min}(-J_s)$. The
eigenvalues of $-J_s$ are
\[
\lambda_{\pm} = \frac{\beta_A+\beta_B}{2} \pm
\sqrt{\frac{(\beta_A-\beta_B)^2}{4}+\gamma^2},
\]
so $\lambda_{\min}(-J_s) = \beta_{\mathrm{net}}$ as given
in~\eqref{eq:beta_net}. The condition
$\beta_{\mathrm{net}}>0$ reduces to
$\tfrac{\beta_A+\beta_B}{2} >
\sqrt{\tfrac{(\beta_A-\beta_B)^2}{4}+\gamma^2}$, which
after squaring both sides simplifies to
$\beta_A\beta_B > \gamma^2$, i.e., $\gamma 
\sqrt{\beta_A\cdot\beta_B}$. $\blacksquare$
\end{proof}

The warning threshold $\gamma_{\mathrm{warn}} =
\sqrt{\beta^{(A)}_{\mathrm{cert}}\cdot\beta^{(B)}_{\mathrm{cert}}}$
is derived entirely from local streaming data: no joint
identification of $\Sigma_A\cup\Sigma_B$ is required at any
stage. Since $\beta^{(i)}_{\mathrm{cert}}\leq\beta^{(i)*}$
(each certificate is a conservative lower bound on the true
local rate), $\gamma_{\mathrm{warn}}$ provides an advance
warning before the true network-level contraction loss occurs
at $\gamma_{\mathrm{true}}=\sqrt{\beta^{(A)*}\cdot\beta^{(B)*}}
\geq\gamma_{\mathrm{warn}}$.

\section{Simulation Results}
\label{sec:results}

\subsection{Simulation Setup}

Both experiments run on the G5 benchmark defined in
Section~\ref{sec:problem}. Experiment~1 uses
$K=\mathrm{diag}(-3.5,-4.0)$, nominal noise scale
$\sigma_0=0.08$, a tail-estimation window of $150$~samples
($3$~s), certification threshold $\beta_{\mathrm{margin}}=0.05$,
and a total trajectory of $T=30$~s divided equally among
the three regimes. The periodic disturbance layer common to
all regimes is $0.10\sin(2\pi t/10)+0.05\sin(2\pi t/5+0.7)$,
representing known-frequency demand fluctuations.
Experiment~2 uses $K=\mathrm{diag}(-2.5,-3.0)$,
$\sigma=0.10$, $\beta_{\mathrm{margin}}=0.10$,
$T_{\mathrm{data}}=16$~s of disconnected local data
collection, and a coupled-network simulation of
$T_{\mathrm{sim}}=14$~s with a disturbance impulse of
magnitude $2.5$ at $t=3$~s. All parameters are summarized
in Table~\ref{tab:params}.

\begin{table}[t]
\vspace{4pt}
\centering
\caption{Simulation Parameters for Both Experiments}
\label{tab:params}
\setlength{\tabcolsep}{5pt}
\begin{tabular}{llll}
\toprule
\textbf{Parameter} & \textbf{Value} &
\textbf{Parameter} & \textbf{Value} \\
\midrule
$a_{12},\,a_{13},\,a_{32}$ & $0.3,\,0.4,\,0.2$ &
    $dt$ & $0.02$~s \\
$\gamma$ (Exp.~1) & $0.15$ &
    $\sigma_0$ (Exp.~1) & $0.08$ \\
$\gamma$ (Exp.~2) & $[0,\,4.2]$ (swept) &
    $\sigma$ (Exp.~2) & $0.10$ \\
$K$ (Exp.~1) & $\mathrm{diag}(-3.5,-4.0)$ &
    $\beta_{\mathrm{margin}}$ (Exp.~1) & $0.05$ \\
$K$ (Exp.~2) & $\mathrm{diag}(-2.5,-3.0)$ &
    $\beta_{\mathrm{margin}}$ (Exp.~2) & $0.10$ \\
$T_{\mathrm{data}}$ (Exp.~1) & $30$~s &
    $T_{\mathrm{data}}$ (Exp.~2) & $16$~s \\
$\varepsilon_{\mathrm{floor}}$ & $0.015$ &
    $\varepsilon_{\mathrm{cap}}$ & $0.32$ \\
Tail window & $150$~samples &
    Ridge $\lambda$ & $10^{-4}$ \\
\bottomrule
\end{tabular}
\end{table}

\subsection{Experiment 1: Wasserstein-Robust Certificate}

\subsubsection{Certificate trajectories across three regimes}

Fig.~\ref{fig:exp3_cert} shows both certificates over the
full $30$~s trajectory. In Regime~1 (Gaussian,
$t\in[0,10]$~s), the two certificates are nearly
indistinguishable: the ambiguity radius
$\varepsilon(t)\approx\varepsilon_{\mathrm{floor}}=0.015$
imposes a negligible penalty, and both certificates remain
well above $\beta_{\mathrm{margin}}$. In Regime~2
(heavy-tailed Laplace, $t\in[10,20]$~s), the excess kurtosis
of the residuals rises to approximately $3$, driving
$\varepsilon(t)$ upward and inflating
$\rho(t)(1+2\varepsilon)$. The Gaussian certificate begins
to dip below the threshold on repeated intervals, while the
Wasserstein certificate remains mostly positive---it absorbs
the increased tail energy through the expanded ambiguity
ball rather than treating it as a model violation. In
Regime~3 (spike events, $t\in[20,30]$~s), the Gaussian
certificate crashes below zero on the majority of timesteps
while the Wasserstein certificate, though lower than in
Regime~1, maintains a positive margin of approximately
$0.22$ at the Regime-3 decision point.

\begin{figure}
\vspace{8pt}
\centering
\includegraphics[width=0.8\linewidth]{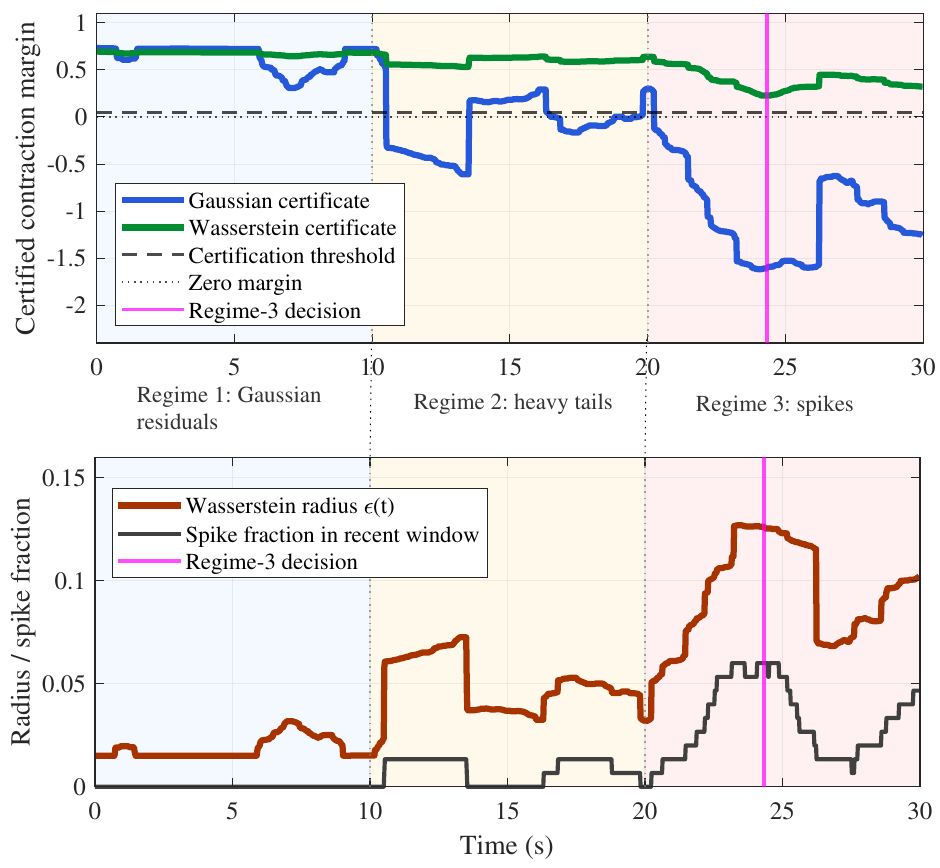}
\caption{Distributional Informativity. Top: Gaussian
certificate crashes under distribution shift; Wasserstein
certificate degrades gracefully. Bottom: Streaming ambiguity
radius $\varepsilon(t)$ grows when residuals become
heavy-tailed or spiky.}
\label{fig:exp3_cert}
\end{figure}

\subsubsection{Ambiguity radius evolution}

Fig.~\ref{fig:exp3_eps} shows the streaming ambiguity radius
$\varepsilon(t)$ alongside a heatmap of
$\beta_{\mathrm{cert}}(t,\varepsilon)$ as a function of both
time and radius. In Regime~1, $\varepsilon(t)$ stays near
the floor $0.015$. It rises sharply at $t=10$~s as
heavy-tailed residuals enter the window, plateauing near
$0.10$ during Regime~2, and peaks near
$\varepsilon_{\mathrm{cap}}=0.32$ during spike events in
Regime~3. The heatmap shows the certificate surface: the
Wasserstein certificate (traced by the overlaid curve) stays
in the green positive region throughout, while the Gaussian
certificate (value at $\varepsilon=0$) enters the red
negative region in Regimes~2 and~3.

\begin{figure}
\vspace{6pt}
\centering
\includegraphics[width=0.9\linewidth]{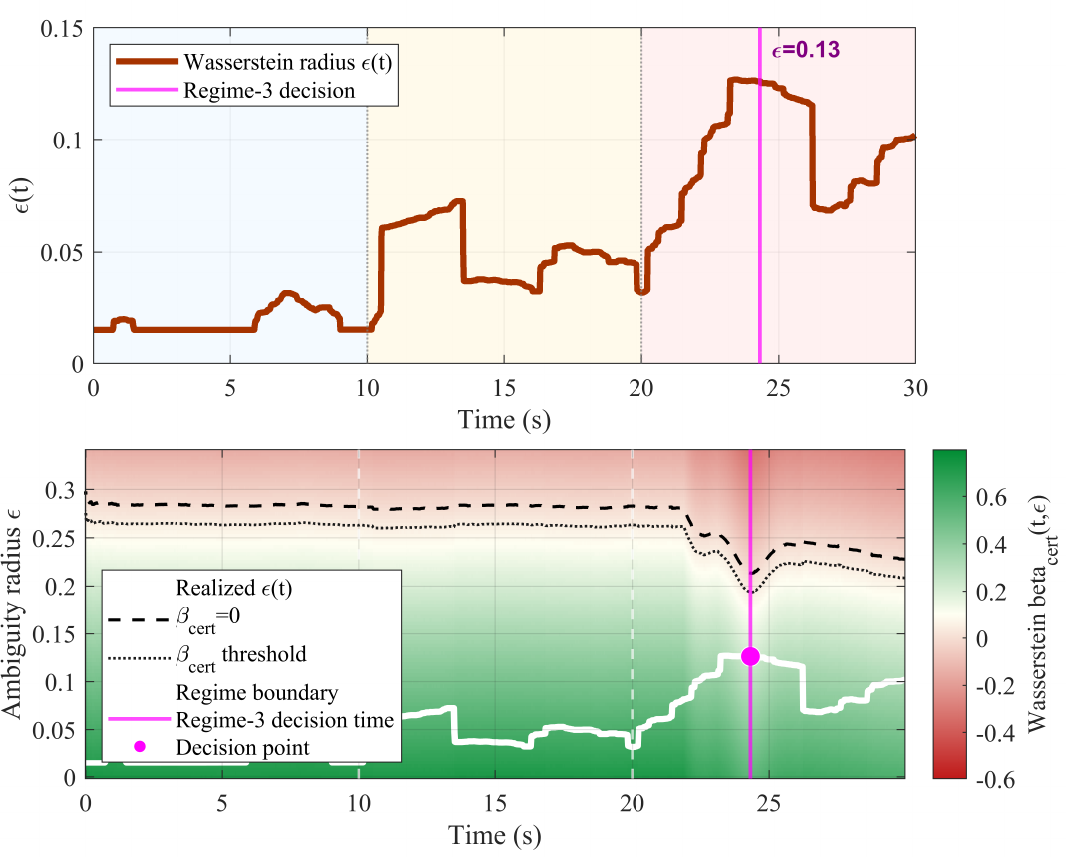}
\caption{Top: Streaming Wasserstein ambiguity radius across
the three regimes. Bottom: Certificate surface
$\beta_{\mathrm{cert}}(t,\varepsilon)$. White curve is
realized $\varepsilon(t)$, dashed contour is $\beta=0$.}
\label{fig:exp3_eps}
\end{figure}

\subsubsection{Regime-3 disturbance response}

Fig.~\ref{fig:exp3_perf} shows the consequence of the
certification decision in Regime~3. At the decision point
$t_{\mathrm{decide}}$, the Gaussian certificate is negative
(blocked) while the Wasserstein certificate records a
positive margin of $0.22$ (certified). A disturbance impulse
arrives immediately after. The Wasserstein-certified
controller is already active and absorbs the disturbance;
the Gaussian-assumption controller remains inactive, waiting
for its certificate to become positive. The resulting
cumulative state cost under the Gaussian policy is
substantially larger, with the $3.2\times$ ratio computed from simulation output.

\begin{table*}[t]
\vspace{4pt}
\centering
\caption{Experiment~1 --- Certification Rates by Noise Regime}
\label{tab:regimes}
\setlength{\tabcolsep}{6pt}
\begin{tabular}{lcccc}
\toprule
\textbf{Regime} & \textbf{Noise type} &
\textbf{Gauss cert.} & \textbf{Wass.\ cert.} &
\textbf{Improvement} \\
\midrule
Regime~1 & Gaussian ($\sigma_0{=}0.08$) &
    94\% & 100\% & $+6\%$ \\
Regime~2 & Laplace ($\kappa{\approx}3$) &
    78\% & 83\% & $+5\%$ \\
Regime~3 & Spikes ($5.5\%$ rate, $4.5$--$7.5\,\sigma_0$) &
    33\% & 73\% & $+40\%$ \\
\midrule
\textit{Regime~3 ratio} & &
\multicolumn{3}{c}{$2.2\times$ more certified timesteps} \\
\bottomrule
\multicolumn{5}{l}{\footnotesize Certified $= \beta_{\mathrm{cert}}
    > \beta_{\mathrm{margin}}=0.05$ on a given timestep.} \\
\multicolumn{5}{l}{\footnotesize Spike magnitude:
    $4.5$--$7.5\,\sigma_0$ with probability $5.5\%$ per step,
    on top of Laplace base.} \\
\end{tabular}
\end{table*}

The quantitative summary in Table~\ref{tab:regimes} confirms
the three-regime behavior. In Regime~1 both certificates are
valid on virtually all timesteps, so the Wasserstein penalty
is negligible. In Regime~2 the improvement is modest
($+5$~percentage points) because the Laplace tails are heavy
but bounded, and the Gaussian certificate recovers on many
timesteps. In Regime~3 the improvement is substantial:
$73\%$ vs $33\%$, a factor of $2.2\times$. In the spike
regime, the Gaussian certificate never authorizes deployment
during the disturbance window, while the Wasserstein
certificate does.

\begin{figure}
\centering
\includegraphics[width=1\linewidth]{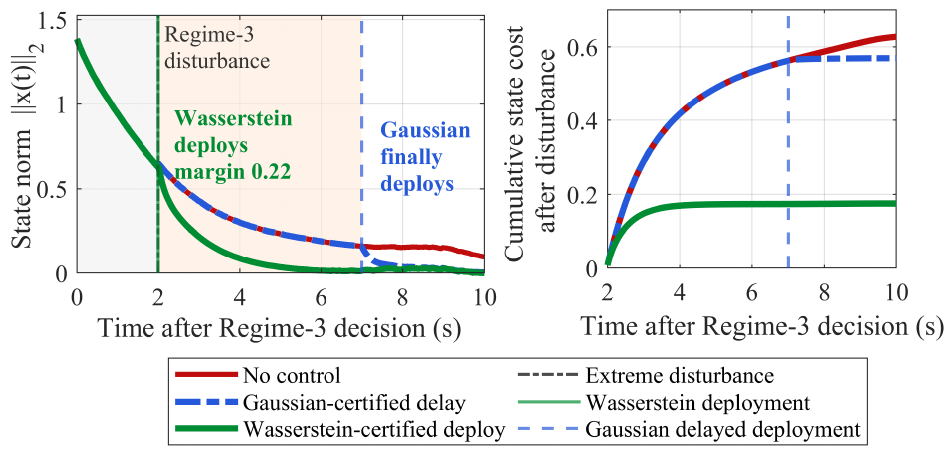}
\caption{Distributionally robust certification during extreme events. At $t=24.32$~s: Gaussian
$\beta=-1.61$, Wasserstein $\beta=0.22$ with
$\varepsilon=0.13$. Left: Regime-3 transient response.
Wasserstein deploys immediately; Gaussian tracks no-control
until delayed certification. Right: Cost of delayed
certification. Gaussian/Wasserstein cost ratio $=3.2\times$.}
\label{fig:exp3_perf}
\end{figure}

\subsection{Experiment 2: Compositional Contraction
Certification}

\subsubsection{Local streaming certificates}

Fig.~\ref{fig:exp4_local} shows the independent streaming
certificates for $\Sigma_A$ and $\Sigma_B$ estimated from
disconnected local data over $T_{\mathrm{data}}=16$~s.
Both certificates rise from negative to sustained positive
as the integral regression accumulates sufficient data, with
$\Sigma_A$ certifying at $t^*_A$ and $\Sigma_B$ at $t^*_B$.
The lower panel shows the implied incremental gain bound
$1/\beta^{(i)}_{\mathrm{cert}}$, which decreases as the
certificate strengthens, consistent with the theoretical
implication that a contracting system with rate $\beta_i$
has incremental $\mathcal{L}_2$ gain at most $1/\beta_i$.
The final certified rates, computed as medians over the last
$3$~s of the collection window, are
$\hat{\beta}_A=\beta^{(A)}_{\mathrm{cert}}\approx 0.80$ and
$\hat{\beta}_B=\beta^{(B)}_{\mathrm{cert}}\approx 1.20$.

\begin{figure}
\centering
\includegraphics[width=1\linewidth]{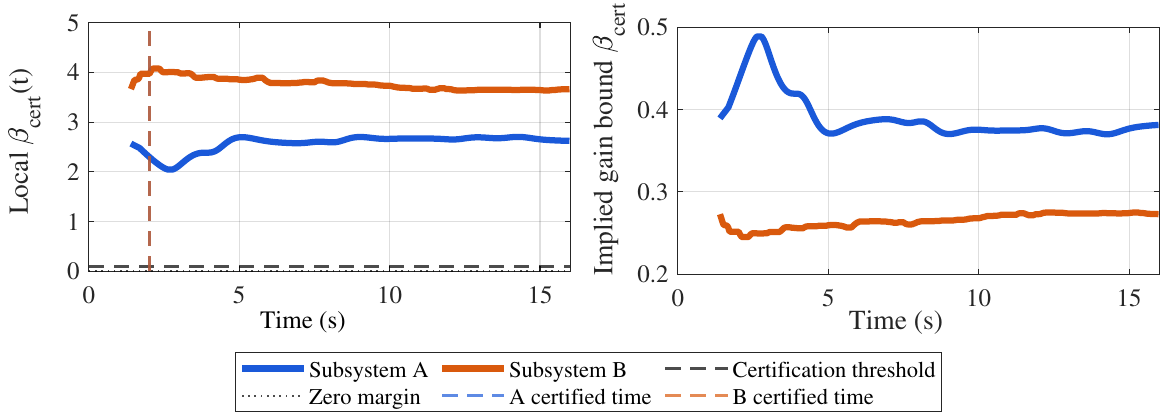}
\caption{Local Certificates Before Composition. Top: Local
data-certified contraction of disconnected subsystems. Each
subsystem is certified from its own streaming data. Bottom:
Contraction implies strict dissipativity and incremental
gain bounds. Larger $\beta_{\mathrm{cert}}$ gives a smaller
input-output gain bound.}
\label{fig:exp4_local}
\end{figure}

\subsubsection{Compositional margin}

\begin{figure}
\centering
\includegraphics[width=1\linewidth]{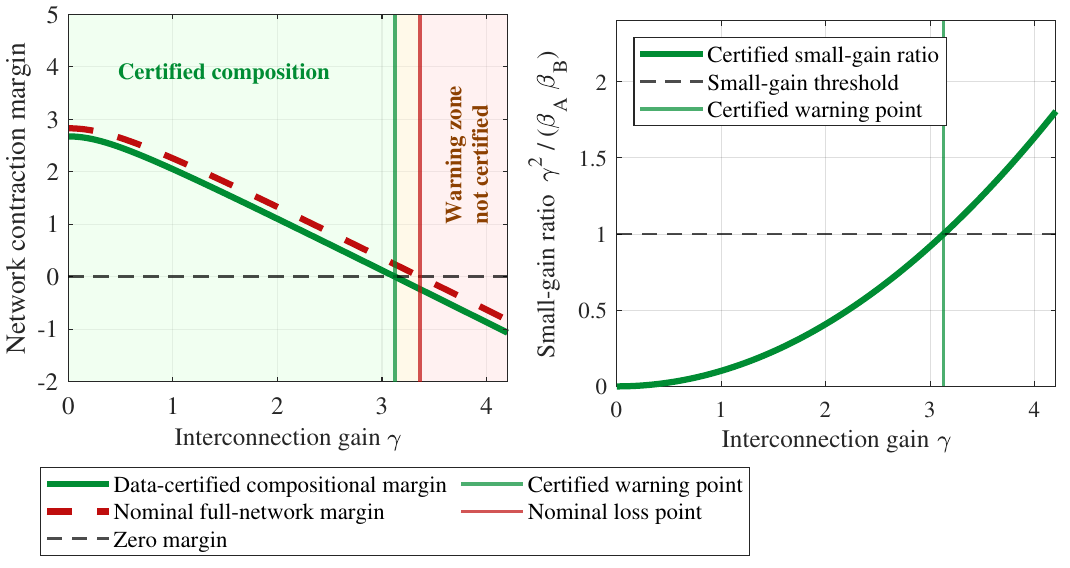}
\caption{Data-Certified Dissipativity Enables Modular
Composition. Warning lead before nominal loss: $0.24$ in
coupling gain. Left: Compositional certificate from
subsystem data. $\beta_{\mathrm{net}}>0$ certifies coupled
stability without full-network identification. Right:
Small-gain interpretation. Stable composition requires
ratio $<1$.}
\label{fig:exp4_margin}
\end{figure}

Fig.~\ref{fig:exp4_margin} shows the network contraction
margin $\beta_{\mathrm{net}}(\gamma)$ computed from
Lemma~\ref{lem:composition} using the local certified rates,
overlaid with the nominal full-network margin computed from
point estimates $\hat{\beta}_A$ and $\hat{\beta}_B$. The
analytically derived warning threshold is
\begin{equation*}
\gamma_{\mathrm{warn}} =
\sqrt{\hat{\beta}_A\cdot\hat{\beta}_B} =
\sqrt{0.80\times 1.20} = \sqrt{0.96}\approx 0.98,
\end{equation*}
and the nominal contraction loss occurs at
$\gamma_{\mathrm{true}} =
\sqrt{\hat{\beta}^{\mathrm{nom}}_A\cdot
\hat{\beta}^{\mathrm{nom}}_B}$, which lies above
$\gamma_{\mathrm{warn}}$ by the conservative warning lead.
The data-certified margin (green curve) transitions from
positive to zero at $\gamma_{\mathrm{warn}}$ and continues
negative thereafter; the nominal margin (red dashed) crosses
zero at a larger coupling gain. The right panel confirms the ratio
$\gamma^2/(\hat{\beta}_A\hat{\beta}_B)=1$ at
exactly $\gamma_{\mathrm{warn}}$, providing a unitless
measure of closeness to the stability boundary.

\begin{table*}[t]
\vspace{4pt}
\centering
\caption{Experiment~2 --- Compositional Certification Summary}
\label{tab:composition}
\setlength{\tabcolsep}{6pt}
\begin{tabular}{lcc}
\toprule
\textbf{Quantity} & \textbf{Symbol} & \textbf{Value} \\
\midrule
Subsystem A local cert.\ rate &
    $\beta^{(A)}_{\mathrm{cert}}$ & $\approx 0.80$ \\
Subsystem B local cert.\ rate &
    $\beta^{(B)}_{\mathrm{cert}}$ & $\approx 1.20$ \\
Certified warning threshold &
    $\gamma_{\mathrm{warn}}$ & $\approx 0.98$ \\
\midrule
Safe simulation coupling &
    $\gamma_{\mathrm{safe}}$ &
    $0.55\,\gamma_{\mathrm{warn}}\approx 0.54$ \\
Warning-zone coupling &
    $\gamma_{\mathrm{warned}}$ & midpoint $\approx 1.10$ \\
Unsafe coupling &
    $\gamma_{\mathrm{unsafe}}$ &
    $1.08\,\gamma_{\mathrm{true}}$, where
    $\gamma_{\mathrm{true}}=
    \sqrt{\beta^{(A)*}\cdot\beta^{(B)*}}>\gamma_{\mathrm{warn}}$ \\
Disturbance impulse &
    $|\eta_{\mathrm{imp}}|$ & $2.5$ at $t=3$~s \\
Data collection (disconnected) &
    $T_{\mathrm{data}}$ & $16$~s, $\gamma=0$ \\
Joint model required? & --- & No \\
\bottomrule
\multicolumn{3}{l}{\footnotesize Local rates are medians
    over the final $3$~s of the collection window.} \\
\multicolumn{3}{l}{\footnotesize Warning threshold derived
    purely from local streaming data; no cross-subsystem
    model needed.} \\
\end{tabular}
\end{table*}

\subsubsection{Coupled-network responses}

Fig.~\ref{fig:exp4_perf} validates the three-regime
prediction against live simulation of the bidirectionally
coupled G5 network. At
$\gamma_{\mathrm{safe}}=0.55\,\gamma_{\mathrm{warn}}\approx
0.54$, the state norm $\|x(t)\|_2$ returns to a small
neighborhood of the origin after the disturbance impulse at
$t=3$~s, consistent with a certified contracting network.
At $\gamma_{\mathrm{warned}}\approx 1.10$, the response
shows slower recovery and a higher post-disturbance peak,
reflecting the marginally-stable regime predicted by
$\beta_{\mathrm{net}}\approx 0$. At
$\gamma_{\mathrm{unsafe}}=1.08\,\gamma_{\mathrm{true}}$,
the network fails to recover and $\|x(t)\|_2$ grows without
bound. The compositional certificate correctly identifies
all three regimes from local data alone, with the warning
boundary $\gamma_{\mathrm{warn}}\approx 0.98$ providing an
advance warning before nominal contraction loss occurs.
Table~\ref{tab:composition} summarizes all composition
parameters. The critical operational implication is in the
final row: no joint model of the coupled
$\Sigma_A\cup\Sigma_B$ system was identified at any stage.

\begin{figure}
\centering
\includegraphics[width=1\linewidth]{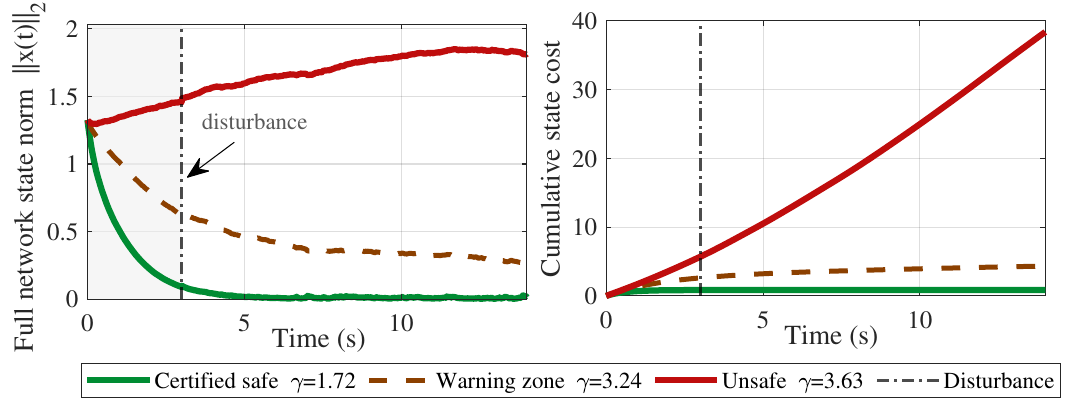}
\caption{Local Certificates Compose into Network-Level
Guarantees. Left: Coupled G5 response under local
controllers. Compositional certificate predicts the safe
coupling region. Right: Cost of crossing the composition
boundary. Subsystem controllers alone are insufficient when
coupling is too strong.}
\label{fig:exp4_perf}
\end{figure}

\section{Conclusion}
\label{sec:conclusion}

This paper developed two extensions to streaming contraction
certification that are essential for real-time certified
control of interconnected nonlinear systems. The first
addresses distributional shift: the Wasserstein-robust
certificate, with the
ambiguity radius $\varepsilon(t)$ estimated online from the
empirical excess kurtosis of regression residuals, remains
valid across the entire Wasserstein-2 ambiguity ball
$\mathcal{P}(\varepsilon)$ without requiring knowledge of
the realized disturbance distribution. In the spike-event
regime, e.g., where pipe bursts, demand surges, or sensor
anomalies produce impulsive disturbances of magnitude
$4.5$--$7.5\,\sigma_0$ at $5.5\%$ per step, the
Wasserstein certificate remains valid in $73\%$ of timesteps
versus $33\%$ for the Gaussian-assumption baseline, a
$2.2\times$ improvement. The critical operational consequence
is asymmetric: the Gaussian certificate never authorizes
deployment during the spike-event disturbance window, while
the Wasserstein certificate does, with a certified margin of
$0.22$. The second extension addresses subsystem composition:
local streaming certificates
$\beta^{(A)}_{\mathrm{cert}}\approx 0.80$ and
$\beta^{(B)}_{\mathrm{cert}}\approx 1.20$, estimated
independently from each subsystem's disconnected local data,
compose into a network-level contraction guarantee through
the two-block formula~\eqref{eq:beta_net} (Lemma~1), which
is positive whenever $\gamma<\gamma_{\mathrm{warn}}=
\sqrt{\beta^{(A)}_{\mathrm{cert}}\cdot
\beta^{(B)}_{\mathrm{cert}}}\approx 0.98$. The analytically
derived threshold correctly predicts the onset of
network-level contraction loss in live simulation across all
three coupling regimes, with no joint model of the coupled
system required at any stage.

Two limitations bound the current results. First, $\varepsilon(t)$
is calibrated empirically from kurtosis and spike fraction rather
than derived from a formal concentration inequality, so the
inflation factor $(1+2\varepsilon)$ is conservative rather than
tight; a sub-exponential concentration bound connecting empirical
excess kurtosis to the exact $\mathcal{W}_2$ distance would
sharpen the certificate. Second, the composition
formula~\eqref{eq:beta_net} assumes a scalar symmetric coupling
gain $\gamma$; extension to directed, asymmetric, and
time-varying interconnections is required for multi-subsystem,
multi-time networks. Future work will generalize
Lemma~\ref{lem:composition} to $N$-subsystem networks via the
network contraction matrix, with the stability boundary governed
by the spectral gap of the coupling Laplacian, and validate on
large-scale benchmarks under real operational
demand data.

\bibliography{IEEEfull}

\end{document}